\documentclass{llncs}
\usepackage[T1]{fontenc}
\usepackage[utf8]{inputenc}

\title{Simple Games versus Weighted Voting Games: Bounding the Critical Threshold Value\thanks{A partial answer to the conjecture of Freixas and Kurz appeared, together with some results of this paper, in an extended abstract published in the proceedings of SAGT 2018~\cite{HKKP18}.}}

\author{Frits Hof\inst{1} \and Walter Kern\inst{1} \and Sascha Kurz\inst{2} \and Kanstantsin Pashkovich\inst{3}\and Dani\"el~Paulusma\inst{4}}

\institute{University of Twente, The Netherlands, \texttt{f.hof@home.nl,w.kern@math.utwente.nl} 
\and  University of Bayreuth, Germany \texttt{sascha.kurz@uni-bayreuth.de}
\and University of Waterloo, Canada \texttt{kpashkovich@uwaterloo.ca}
\and Durham University, United Kingdom \texttt{daniel.paulusma@durham.ac.uk}}

\usepackage{amsmath,amssymb}
\usepackage{hyperref}
\usepackage{graphicx}
\usepackage{enumerate}
\usepackage{tikz}
\usepackage{boxedminipage}
\usetikzlibrary{arrows,shapes,calc}
\hypersetup{hypertexnames=false}

\renewcommand{\le}{\leqslant}
\renewcommand{\ge}{\geqslant}
\renewcommand{\leq}{\leqslant}
\renewcommand{\geq}{\geqslant}
\renewcommand{\phi}{\varphi}

\newcommand{\R}{\mathbb{R}}

\newcommand{\W}{\mathcal{W}}
\newcommand{\Lo}{\mathcal{L}}

\newcommand{\0}{\boldsymbol{0}}
\newcommand{\1}{\boldsymbol{1}}

\newcommand{\conv}{\mathrm{conv}}

\newcommand{\ip}[2]{\langle{#1},{#2}\rangle}

\newcommand{\NP}{{\sf NP}}

\newtheorem{thm}{Theorem}
\newtheorem{conj}{Conjecture}
\newtheorem{rem}[thm]{Remark}
\newtheorem{ex}[thm]{Example}

\begin{document}
\maketitle

\begin{abstract}
\noindent
A simple game $(N,v)$ is given by a set $N$ of $n$ players and a partition of~$2^N$ into  a set~$\mathcal{L}$ of losing coalitions~$L$ with value 
$v(L)=0$ that is closed under taking subsets and a set $\mathcal{W}$ of winning coalitions $W$ with $v(W)=1$.  Simple games with 
$\alpha= \min_{p\geq 0}\max_{W\in {\cal W}, L\in {\cal L}} \frac{p(L)}{p(W)}<1$ are exactly the weighted voting games. We show that $\alpha\leq \frac{1}{4}n$ for every simple game $(N,v)$, confirming the conjecture of Freixas and Kurz (IJGT, 2014). For complete simple games, Freixas and Kurz conjectured that $\alpha=O(\sqrt{n})$. We prove this conjecture up to a $\ln n$ factor. We also prove that for graphic simple games, that is, simple games in which every minimal winning coalition has size~2, computing $\alpha$ is \NP-hard, but polynomial-time solvable if the underlying graph is bipartite. Moreover, we show that for every graphic simple game, deciding if $\alpha<a$ is polynomial-time solvable for every fixed $a>0$.
\end{abstract}

\section{Introduction}\label{sec_introduction}

Cooperative Game Theory provides a mathematical framework for capturing situations where subsets of agents  may form a coalition in order 
to obtain some collective profit or share some collective cost. Formally, 
a \emph{cooperative game  (with transferable utilities)} consists of a pair~$(N,v)$, where $N$ is a set of $n$ agents called \emph{players} and
$v: 2^N\to \R_+$ is a \emph{value function} that satisfies $v(\varnothing) = 0$.
In our context, the value $v(S)$ of a \emph{coalition} $S\subseteq N$ represents the profit for $S$ if all players in $S$ choose to collaborate with (only) each other. 
The central problem in cooperative game theory is to allocate the total profit $v(N)$ of the \emph{grand coalition}~$N$ to the 
individual players $i \in N$ in a {\lq\lq}fair{\rq\rq} way. To this end various \emph{solution concepts} such as the core, Shapley value or nucleolus 
have been designed; see~\cite{peters2008games} for an overview. 
For example, core solutions try to allocate the total profit such that every coalition
$S \subseteq N$ gets at least $v(S)$. This is of course not always possible, that is, the core might be empty. This leads to related questions
like: ``How much do we need to spend in total if we want to give at least $v(S)$ to each coalition $S \subseteq N$?''. In the specific case of simple
games (\emph{cf.} below) where $v$ takes only values $0$ and $1$, classifying coalitions into ``losing'' and ``winning'' coalitions, one may also ask:
``How much do we have to give in the worst case to a losing coalition if we want to give at least $v(S)=1$ to each winning coalition?''

As mentioned above, we study simple games.
 Simple games form a classical class of games, which are well studied;  see also the book of Taylor and Zwicker~\cite{taylor1999simple}.
The notion of being simple means that every coalition either has some equal amount of power or no power at all.  Formally, a cooperative game $(N,v)$ is \emph{simple} if  $v$ is a monotone 0--1 function with $v(\varnothing)=0$ and $v(N)=1$, so
$v(S)\in \{0,1\}$ for all $S\subseteq N$ and $v(S)\leq v(T)$ whenever $S\subseteq T$. 
 In other words, if $(N,v)$ is simple, then
there is a set $\mathcal{W} \subseteq 2^N$ of \emph{winning coalitions} $W$ that have value $v(W)=1$ and
a set $\mathcal{L} \subseteq 2^N$of \emph{losing coalitions} $L$ that have value $v(L)=0$. 
Note that $N\in {\cal W}$, $\varnothing \in {\cal L}$ and ${\cal W}\cup {\cal L}=2^N$. The monotonicity of~$v$ implies that
subsets of losing coalitions
are losing and supersets of winning coalitions are winning. A winning coalition~$W$ is {\it minimal} if every proper subset 
of $W$ is losing, and a losing coalition~$L$ is {\it maximal} if every proper superset of $L$ is winning.

A simple game is a \emph{weighted voting game} if there exists a payoff vector $p \in \R_+^n$  
such that a coalition~$S$ is winning if $p(S)\geq 1$ and losing if $p(S)<1$. 
Weighted voting games are also known as {\it weighted majority games} and form one of the most popular classes of simple games.

However, it is easy to construct simple games that are not weighted voting games. We give an example below, but in fact there are many important simple games that are not weighted voting games, and the relationship between weighted voting games and simple games is 
not yet fully understood. Therefore, Gvozdeva, Hemaspaandra, and Slinko~\cite{gvozdeva2013three} introduced a parameter $\alpha$, called the 
{\it critical threshold value}, to measure the {\lq\lq}distance{\rq\rq} of a 
simple game to the class of weighted voting games:
\begin{equation} \label{eq_minmax}
\alpha ~= ~\alpha(N,v)~=\min_{p \ge 0}~ \max_{\substack{W\in \mathcal{W}\\L\in \mathcal{L}}}~ \frac{p(L)}{p(W)}\,.
\end{equation}
\noindent
A simple game $(N,v)$ is a weighted voting game if and only if $\alpha<1$.
This follows from observing that each optimal solution $p$ of 
(\ref{eq_minmax}) can be scaled to satisfy $p(W)\ge 1$ for all winning coalitions $W$. The scaling enables us to reformulate the critical threshold value as follows:
\[
\alpha=\alpha(N,v)=\min_{p\in Q(\W)} \max_{L\in \Lo} p(L)\,,
\]
where
\[
Q(\W)=\{p\in\R^N\,|\, p(W) \ge \1 \text{  for  every } W\in\W,\, p\ge \0\}\,.
\]

The following concrete example of a simple game $(N,v)$ that is not a weighted voting game and that has in fact a large value of $\alpha$ was given  in~\cite{paper_alpha_roughly_weighted}:

\begin{ex}\label{example}
Let $N=\{1, \dots, n\}$ for some even integer~$n\geq 4$, and let the minimal winning coalitions be the pairs $\{1,2\}, \{2,3\}, \dots \{n-1, n\}, \{n,1\}$.
Then
\[
Q(\W)=\{p\in\R^N\,|\, p_1+p_2\geq 1,\, p_2+p_3\geq 1,\,\ldots,\, p_n+p_1\geq 1,\, p\ge \0\}\,.
\]
This means that $p(N)\geq \frac{1}{2}n$ for every $p\in Q(\W)$. 
Then, for every $p\in Q(\W)$ and for at least one of the two loosing coalitions $\{2,4,6,\dots, n\}$ or $\{1,3,5, \dots, n-1\}$, we have  
$p(L)\geq \frac{1}{4}n$, showing that $\alpha \ge \frac{1}{4}n$. On the other hand, it is easily seen that $p \equiv \frac{1}{2}$ satisfies $p(W) \ge 1$
for all  winning coalitions and $p(L) \le \frac{1}{4}n$ for all losing coalitions, showing that $\alpha \le \frac{1}{4}n$. Thus 
 $\alpha=\frac{1}{4}n$.
 \end{ex}
 
This example led the authors of~\cite{paper_alpha_roughly_weighted} to the following conjecture:

\begin{conj}[\cite{paper_alpha_roughly_weighted}]\label{conj}
For every simple game $(N,v)$, it holds that $\alpha \le \frac{1}{4}n$.
\end{conj}

\subsubsection*{Our Results.}
Section~\ref{sec_proof} contains our main result.
In this section we reformulate and strengthen Conjecture~\ref{conj} and then we prove the obtained strengthening. 

In  Section~\ref{sec_complete} we consider a subclass of simple games based on a natural desirability order~\cite{Is56}. A simple game $(N,v)$ is \emph{complete} if the players can be ordered by a complete, transitive ordering $\succeq$, say, $1 \succeq 2 \succeq \dots \succeq n$, indicating that higher ranked players have more "power" than lower ranked players.
More precisely, $i \succeq j$ means that $v(S\cup i) \ge v(S \cup j)$ for any coalition $S \subseteq N\backslash \{i,j\}$. The class of complete simple games properly contains all weighted voting games~\cite{FP08}. For complete simple games, we show an asymptotically upper bound on $\alpha$,  namely $\alpha=O(\sqrt{n}\ln n)$.
This bound matches,  up to a $\ln n$ factor, the lower bound of $\Omega(\sqrt{n})$ in~\cite{paper_alpha_roughly_weighted}, where the bound $\Omega(\sqrt{n})$ is
conjectured to be tight in~\cite{paper_alpha_roughly_weighted}. Intuitively, complete simple games are much closer to weighted voting games than 
arbitrary simple games. So, from this perspective, our result seems to support the hypothesis that $\alpha$ is indeed a sensible measure for the 
distance to weighted voting games. 

In Section \ref{sec_algo} we discuss some algorithmic and complexity issues. We focus on instances where all minimal winning coalitions have size~$2$.  We say that such simple games are {\it graphic}, as they can conveniently be described by a graph $G=(N,E)$
with vertex set $N$ and edge set $E~=~\{ij~|~\{i,j\}~ \text{is winning} \}$.
For graphic simple games we show that computing~$\alpha$ 
is \NP-hard in general, but
 polynomial-time solvable if the underlying graph $G=(N,E)$ is bipartite, or if $\alpha$ is known to be small  (less than a fixed number $a$). 

\medskip
\noindent
{\bf Related Work.}
Due to their practical applications in voting systems, computer operating systems and model resource allocation (see e.g.~\cite{BFJL02,CEW11}), structural and computational 
complexity aspects for solution concepts for weighted voting games have
been thoroughly investigated~\cite{ECJ08,EGGW09,freixas2011complexity,gvozdeva2013three,Pa18a}.

Another way to measure the distance of a simple game to the class of weighted voting games is to use the {\it dimension} of a simple game~\cite{TZ93}, which is the smallest number 
of weighted voting games whose intersection equals a given simple game. 
However, computing the dimension of a simple game is \NP-hard~\cite{DW06}, and
the largest dimension of a simple game with $n$ players is $2^{n-o(n)}$~\cite{high_dimension}. 
Moreover, $\alpha$ may be arbitrarily large for simple games with dimension larger than~1.
Hence there is no direct relation between the two distance measures.
Gvozdeva, Hemaspaandra, and Slinko~\cite{gvozdeva2013three} introduced two other distance parameters as well. One measures the power balance between small and large coalitions. The other one allows multiple thresholds instead of threshold~1 only. 

For graphic simple games, it is natural to take the number of players $n$ as the input size for answering complexity questions, but in general simple games may have different representations. For instance, one can list all minimal winning coalitions or all maximal losing coalitions. Under these two representations the problem of deciding if 
 $\alpha < 1$, that is, if a given simple game is a weighted voting game, is also polynomial-time solvable.
 This follows from results of \cite{hegedus1996geometric,peled1985polynomial}, as shown in~\cite{freixas2011complexity}.
The latter paper also showed that the same result holds if the representation is given by listing all winning coalitions or all losing coalitions. 

As mentioned, a crucial case in our 
study is when the simple game is graphic, that is, defined on some graph $G=(N,E)$.
In the corresponding \emph{matching game} a coalition $S 
\subseteq N$ has value $v(S)$ equal to the maximum size of a matching in the subgraph of $G$ induced by $S$. 
One of the most prominent solution concepts is the \emph{core} of a game, defined by $core(N,v) := \{p \in \R^n ~|~ p(N)=v(N), ~p(S) \ge v(S)~ \forall S  \subseteq N\}$. 
Matching games are not simple games. Yet their core constraints are readily seen to simplify to $p\ge 0$ and $p_i+p_j \ge 1$ for all $ij \in E$.
Classical solution concepts, such as the core and core-related ones like least core,  nucleolus or nucleon are
well studied for matching games, see, for example,~\cite{biro2012matching,bock2015stable,faigle1998,kern2003matching,KKT,SR94}. 

\section{The Proof of the Conjecture}~\label{sec_proof} 
To prove Conjecture~\ref{conj} we reformulate, strengthen and only then verify it. Our approach is inspired by the work of
Abdi, 
Cornu\'{e}jols and  
Lee on identically self-blocking clutters~\cite{A18}.
A coalition~$C\subseteq N$ is called a \emph{cover} of $\W$ if $C$ has at least one common player with every coalition in $\W$. We call the collection of covers of~$\W$ the \emph{blocker} of $\W$ and denote it by $b(\W)$\footnote{Usually, 
the notion of a blocker is defined as the collection of minimal covers, but for simplicity of exposition, we define it as the collection of all covers.} \cite{Edmonds70}. 
We claim that
\[
\Lo=\{ N\setminus C\,|\, C\in b(\W)\}\,.
\]
In order to see this, first suppose that there exists a cover $C\in b({\cal W})$ such that $N\setminus C\notin {\cal L}$. As ${\cal L}\cup {\cal W}=2^N$, this means that $N\setminus C\in {\cal W}$. 
However, as $C$ contains no player from $N\setminus C$, this contradicts our assumption that $C\in b({\cal W})$.  Now suppose that there exists a losing coalition $L\in {\cal L}$ such that $C=N\setminus L$ does not belong to $b({\cal W}$. Then, by definition, there exists a winning coalition $W\in {\cal W}$ with $C\cap W=\emptyset$.  As $C\cap W=\emptyset$, we find that $W\subseteq N\setminus C=L$. Then, by the monotonicity property of simple games, $L$ must be winning as well, a contradiction.

As $\Lo=\{ N\setminus C\,|\, C\in b(\W)\}$, the critical threshold value can be reformulated as follows
\begin{align*}
\alpha=\min_{p\in Q(\W)} \max_{L\in \Lo} p(L)=
 &\min_{p\in Q(\W)} \max_{C\in b(\W)} p(N\setminus C)=\\
&\min_{p\in Q(\W)} \max_{\substack{q\in Q(\W)\\q\in \{0,1\}^N}} \ip{p}{\1-q}\,.
\end{align*}
Here, $\ip{p}{q}$ stands for the scalar product of two vectors $p$ and $q$. To see the last equality, for a cover $C$ we can define a corresponding vector $q\in \{0,1\}^N\cap Q(\W)$ by setting $q_i=1$ if $i\in C$ and $q_i=0$ otherwise. 

\medskip
\noindent
{\bf Conjecture~\ref{conj} (reformulated)}
{\it For a simple game with $n$ players and the collection of winning coalitions $\W$, we have
\[
\min_{p\in Q(\W)} \max_{\substack{q\in Q(\W)\\q\in \{0,1\}^N}} \ip{p}{\1-q}\leq n/4\,.
\]}

\medskip
\noindent
Next, we prove Theorem~\ref{thm_strength}, which is a strengthening of Conjecture~\ref{conj}. For the proof we need the following straightforward remark, which we leave as an exercise. 
Here, we write $\left\lVert p \right\rVert_2=\sqrt{p_1^2+\ldots +p_n^2}$ for a vector $p\in \R^n$.

\begin{rem}\label{rem_min_norm}
Let $P$ be a polyhedron and let $p^\star$ be the optimal solution of the program $\min \{\left\lVert p \right\rVert_2\,|\, p\in P\}$. Then $p^\star$ is an optimal solution of the linear program $\min \{\ip{p^\star}{q}\,|\, q\in P\}$.
\end{rem}

\begin{thm}[Strengthening of Conjecture~\ref{conj}]\label{thm_strength}
For a simple game with $n$ players and the collection of winning coalitions $\W$, we have
\[
\min_{p\in Q(\W)} \max_{q\in Q(\W)} \ip{p}{\1-q}\leq n/4\,.
\]
In particular, if $p^\star$ is the optimal solution for the program
\[
\min \{\left\lVert p \right\rVert_2\,|\, p\in Q(\W)\}\,,
\]
then
\[
 \max_{q\in Q(\W)} \ip{p^\star}{\1-q}\leq n/4\,.
\]
\end{thm}
\begin{proof}

Consider the unique optimal solution $p^\star$ for the program $\min \{\left\lVert p \right\rVert_2\,|\, p\in Q(\W)\}$. By Remark~\ref{rem_min_norm}, $p^\star$ is an optimal solution for the program $\min \{\ip{p^\star}{q}\,|\, q\in Q(\W)\}$. Thus, $p^\star$ is an optimal solution for the program $\max_{q\in Q(\W)} \ip{p^\star}{\1-q}$. Thus, we have
\[
\max_{q\in Q(\W)} \ip{p^\star}{\1-q}=\ip{p^\star}{\1-p^\star}=\frac{n}{4}-\ip{\frac{1}{2}\1-p^\star}{\frac{1}{2}\1-p^\star}\leq \frac{n}{4}\,,
\]
finishing the proof.\qed
\end{proof}

Let us discuss when Conjecture~\ref{conj} provides a tight upper bound for the critical threshold value. The next theorem shows that if the upper bound in Conjecture~\ref{conj} is tight, then this fact can be certified in the same way as in Example~\ref{example}.

\begin{thm}
For a simple game with $n$ players and the collection of winning coalitions $\W$ and the collection of losing coalitions $\Lo$, we have
\[
\alpha=\min_{p\in Q(\W)} \max_{L\in \Lo} p(L)= n/4
\]
if and only if $\frac{2}{n}\1$ lies in the convex hull of the characteristic vectors of winning coalitions and $\frac{1}{2}\1$ lies in the convex hull of the characteristic vectors of losing coalitions.
\end{thm}
\begin{proof}
Clearly, if $\frac{2}{n}\1$ lies in the convex hull of the characteristic vectors of winning coalitions and $\frac{1}{2}\1$ lies in the convex hull of the characteristic vectors of losing coalitions, then for every $p\in Q(\W)$ we have
\[
\max_{L\in \Lo} p(L)\geq \ip{p}{\frac{1}{2}\1}= \frac{n}{4}\ip{p}{\frac{2}{n}\1}\geq \frac{n}{4}\,,
\]
showing that $\alpha\geq n/4$ and hence $\alpha= n/4$ by Theorem~\ref{thm_strength}.

On the other hand, from the proof of Theorem~\ref{thm_strength} we know that if $\alpha=n/4$ then $p^\star=\frac{1}{2}\1$ is an optimal solution for $\min \{\ip{p^\star}{q}\,|\, q\in Q(\W)\}$ with value $n/4$. Let us show that $\frac{2}{n}\1$ lies in the convex hull of the characteristic vectors of winning coalitions. To do that consider an optimal dual  solution $y^\star$ for the program $\min \{\ip{p^\star}{q}\,|\, q\in Q(\W)\}$. Using complementary slackness it is straightforward to show that $\frac{4}{n}y^\star$ provides coefficients of a convex combination of characteristic vectors of winning coalitions, where the convex  combination equals~$\frac{2}{n}\1$.

In the same way as the proof of Theorem~\ref{thm_strength}, we could show that
\[
\alpha\leq\max_{\substack{q\in Q(\W)\\q\in \{0,1\}^N}} \ip{q^\star}{\1-q}=\ip{q^\star}{\1-q^\star}=\frac{n}{4}-\ip{\frac{1}{2}\1-q^\star}{\frac{1}{2}\1-q^\star}\leq \frac{n}{4}\,,
\]
where $q^\star$ is the optimal solution for the program 
\[
\min \{\left\lVert q \right\rVert_2\,|\, q\in \conv\{r\in \{0,1\}^N\,|\,r\in Q(\W)\}\}\,.
\]
Thus, if $\alpha$ equals $n/4$, then $q^\star=\frac{1}{2}\1$ and $\frac{1}{2}\1$ lies in $\conv\{r\in \{0,1\}^N\,|\,r\in Q(\W)\}$. Hence, if $\alpha$ equals $n/4$, then $\1-q^\star=\frac{1}{2}\1$ lies in the convex hull of the characteristic vectors of losing coalitions, finishing the proof. \qed
\end{proof}

\section{Complete Simple Games}\label{sec_complete}

Intuitively, the class 
of complete simple games is {\lq\lq}closer{\rq\rq} to weighted voting games than general simple games. The 
next result quantifies this expectation.

\begin{theorem}
  \label{t-csg}
  For a complete simple game $(N,v)$, it holds that $\alpha\le  \sqrt{n}\ln n$. 
\end{theorem}

\begin{proof}
  Let $N=\{1,\dots,n\}$ be the set of players and assume without loss of generality that $1\succeq 2\succeq\dots \succeq n$.
  Let $k \in N$ be the largest number such that $\{k, \dots, n\}$ is winning.
  For $i=1, \dots, k$, let $s_i$ denote the smallest size of a winning coalition in $\{i, \dots, n\}$.
  Define $p_i:= 1/s_i$ for $i=1, \dots, k$ and $p_i:=p_k$ for $i=k+1,\dots,n$. Thus, obviously, $p_1 \ge \dots \ge p_k = \dots = p_n$. 
  
  Consider a winning coalition $W\subseteq N$ and let $i$ be the first player in $W$ (with respect to $\succeq$).
  If $|W| \le \sqrt{n}$, then $s_i \le |W| \le \sqrt{n}$ and hence $p(W) \ge p_i = \frac{1}{s_i} \ge \frac{1}{\sqrt{n}}$.
  On the other hand, if $|W| > \sqrt{n}$, then $p(W) > \sqrt{n}p_k \ge \sqrt{n}\frac{1}{n} =\frac{1}{\sqrt{n}}$.
  
  For a losing coalition $L \subseteq N$, we conclude that $|L \cap \{1, \dots , i\}| \le s_i-1$ (otherwise $L$ would dominate
  the winning coalition of size $s_i$ in $\{i, \dots ,n\}$). So $p(L)$ is bounded by 
  $\max \sum_{i=1}^{k} x_i\frac{1}{s_i} \text{~~subject to~~} \sum_{j=1}^{i}x_j \le s_i-1,~ i=1, \dots , k$.
  The optimal solution of this maximization problem is $x_1=s_1-1$ and $x_i=s_i-s_{i-1} \text{~for~} i=2,\dots  k$.
  Hence $p(L) \le (s_1-1)\frac{1}{s_1} + (s_2-s_1)\frac{1}{s_2} + \dots + (s_k-s_{k-1})\frac{1}{s_k}
       \le \frac{1}{2} +   \dots + \frac{1}{s_k} \le \ln n$.
  Summarizing, we obtain $p(L)/p(W) \le \sqrt{n} \ln n$.  \qed
   \end{proof}

In \cite{paper_alpha_roughly_weighted} it is conjectured that $\alpha = O(\sqrt{n})$ holds for complete simple games.
In the same paper a lower bound of order $\sqrt{n}$ is given,
as well as specific subclasses of complete simple games for which $\alpha=O(\sqrt{n})$ can be proven.

\section{Algorithmic Aspects}\label{sec_algo}

A fundamental question concerns the complexity of our original problem (\ref{eq_minmax}), i.e., the complexity of computing the critical threshold value of a simple game. For general simple games this depends on how the  game in question is given, and we refer to Section~\ref{sec_introduction} for a discussion.
Here we concentrate on the {\lq\lq}graphic{\rq\rq} case.

\begin{proposition}\label{p-bip}
Computing $\alpha_G$ for bipartite graphs $G$ can be done in polynomial time.
\end{proposition}

\begin{proof}
 Let $P \subseteq \R^n$ be the set of feasible payoffs (satisfying $p \ge 0$ and $p_i+p_j \ge 1$ for $ij \in E$). For $\alpha \in \R$, let
 $P_\alpha := \{p\in P ~|~ p(L) \le \alpha \text{~for all independent ~} L \subseteq N\}$.
Thus $\alpha_G = \min \{\alpha ~|~P_\alpha \neq \varnothing\}$. The separation problem for $P_\alpha$ (for any given $\alpha$) is efficiently solvable.
Given $p\in \R^n$, we can check feasibility and whether $\max \{p(L) ~| ~L \subseteq N \text{~independent} \} \le \alpha $ by solving a 
corresponding maximum weight independent set problem in the bipartite graph $G$. Thus we can, for any given $\alpha \in \R$, apply the ellipsoid
method to either compute some $p \in P_\alpha$ or conclude that $P_\alpha =\varnothing$. Binary search then exhibits the minimum value for which 
$P_\alpha$ is non-empty; binary search works indeed in polynomial time as the optimal $\alpha$ has size polynomially bounded in $n$, which follows from observing that 
\begin{equation}\label{eq_LP}
\alpha = \min \{a ~|~ p_i+p_j \ge 1~ ~\forall ij \in E,~ p(L)-a \le 0 ~ ~\forall L \subseteq N \text{~independent}, p \ge 0\}
\end{equation}
can be computed by solving a linear system of $n$ constraints defining an optimal basic solution of the above linear program. \qed \end{proof}

The proof of Proposition~\ref{p-bip} also applies to other classes of graphs, such as claw-free graphs  (see~\cite{brand2016claw})
in which finding a  weighted maximum independent set is polynomial-time solvable. In general, the problem is \NP-hard.

\begin{proposition}
 Computing $\alpha_G$ for arbitrary graphs $G$ is \NP-hard. 
\end{proposition}

\begin{proof}
Let $G'=(N',E')$ and $G''=(N'',E'')$ be two disjoint copies of a graph $G=(N,E)$ with independence number~$k$. For each $i'\in N'$
 and $j''\in N''$ add  an edge $i'j''$ if and only if $i=j$ or $ij \in E$  and call the resulting graph $G^*= (N^*, E^*)$. 
 We claim that $\alpha_{G^*} = k/2$ (thus computing $\alpha_{G^*}$ is as difficult as computing~$k$).
  
 First note that the independent sets  in $G^*$ are exactly  the sets $L^* \subseteq N^*$ that arise from an independent set $L \subseteq N$ in $G$
 by splitting $L$ into two complementary sets $L_1$ and $L_2$ and defining $L^*:= L_1'\cup L_2''$.
 Hence, $p\equiv \frac{1}{2}$ on $N^*$ yields $\max p(L^*) = k/2$ where the maximum is taken over all independent sets $L^* \subseteq N^*$ in $G^*$. 
 This shows that $\alpha_{G^*} \le k/2$.
 
 Conversely, let $p^*$ be any feasible payoff in $G^*$, that is, $p^* \ge 0$ and $p^*_i+p^*_j \ge 1 $ for all $ij \in E^*$. Let $L \subseteq N$ 
 be a maximum independent set of size $k$ in $G$ and construct $L^*$ by 
 including for each $i \in L$ either $i'$ or $i''$ in $L^*$, whichever has $p$-value at least $\frac{1}{2}$. Then, by construction, $L^*$ is an 
 independent set in $G^*$ with $p^*(L^*) \ge k/2$, showing that $\alpha_{G^*} \ge k/2$. \qed
\end{proof}

Summarizing, for graphic simple games, computing $\alpha_G$ is
as least as hard as computing the size of a maximum independent in~$G$.
For our last result we assume that $a$ is a fixed integer, that is, $a$ is not part of the input.

\begin{proposition}
 For every fixed $a > 0$, it is possible to decide if $\alpha_G \leq a$ in polynomial time for an arbitrary graph $G=(N,E)$.
\end{proposition}

\begin{proof} 
Let $k= 2\lceil a+\epsilon \rceil$ for some $\epsilon >0$. By brute-force, we can check in $O(n^{2k})$ time if $N$ contains 
  $2k$ vertices $\{u_1,\ldots,u_k\}\cup \{v_1,\ldots,v_k\}$ that induce $k$ disjoint copies of $P_2$, that is, paths $P_i=u_iv_i$ of length $2$ for 
  $i=1,\ldots,k$ with no edges joining any two of these paths.
If so, then the condition $p(u_i)+p(v_i)\geq 1$ implies that one of $u_i,v_i$, say $u_i$, must receive a payoff $p(u_i)\geq \frac{1}{2}$,
and hence $U=\{u_1,\ldots,u_k\}$ has $p(U)\geq k/2 > a$. As $U$ is an independent set,  $\alpha(G)>a$.

Now assume that $G$ does not contain $k$ disjoint copies of $P_2$ as an induced subgraph, that is,
 $G$ is $kP_2$-free. For every $s\geq 1$, the number of maximal independent sets in a $sP_2$-free graphs is $n^{O(s)}$ 
due to a result of Balas and Yu~\cite{BY89}. Tsukiyama, Ide, Ariyoshi, and Shirakawa~\cite{TIAS77} show how to enumerate 
all maximal independent sets of a graph $G$ on $n$ vertices and $m$ edges using time $O(nm)$ per independent set. Hence we can find all maximal 
independent 
sets of $G$ and thus solve, in polynomial time, the linear program~(\ref{eq_LP}). Then it remains to check if the 
solution found satisfies $\alpha\leq a$. \qed
\end{proof}
  
\section{Conclusions}\label{sec_conclusion}

We have strengthened and proven the conjecture of~\cite{paper_alpha_roughly_weighted} on simple games (Conjecture~\ref{conj}) and showed a number of computational complexity results for graphic simple games. Moreover, we considered complete simple games and proved a stronger upper bound for this class of games. It remains to tighten the upper bound for complete simple games to $O(\sqrt{n})$ if possible. In order to classify simple games, many more subclasses of simple games have been identified in the literature. Besides the two open problems, no optimal bounds for $\alpha$ are known for other subclasses of simple games, such as \textit{strong}, \textit{proper}, or \textit{constant-sum} games, that is, where $v(S)+v(N\backslash S)\ge 1$, $v(S)+v(N\backslash S)\le 1$,  or $v(S)+v(N\backslash S)= 1$ for all $S\subseteq N$, respectively.\\

\smallskip
\noindent
{\it Acknowledgments.}  The second and fifth author thank P\'eter Bir\'o and Hajo Broersma for fruitful discussions on the topic of the paper. The fourth author thanks Ahmad Abdi  for valuable and helpful discussions.

 \bibliographystyle{abbrv}

\begin{thebibliography}{10}

\bibitem{A18}
A. Abdi.
\newblock Ideal clutters. 
\newblock {\em University of Waterloo}, 2018.

\bibitem{BY89}
E.~Balas and C.~S. Yu.
\newblock On graphs with polynomially solvable maximum-weight clique problem.
\newblock {\em Networks}, 19(2):247--253, 1989.

\bibitem{BFJL02}
J.~M. Bilbao, J.~R.~F. Garc{\'{\i}}a, N.~Jim{\'{e}}nez, and J.~J. L{\'{o}}pez.
\newblock Voting power in the {E}uropean {U}nion enlargement.
\newblock {\em European Journal of Operational Research}, 143(1):181--196,
  2002.

\bibitem{biro2012matching}
P.~Biro, W.~Kern, and D.~Paulusma.
\newblock Computing solutions for matching games.
\newblock {\em International Journal of Game Theory}, 41:75--90, 2012.

\bibitem{bock2015stable}
A.~Bock, K.~Chandrasekaran, J.~K{\"o}nemann, B.~Peis, and L.~Sanit{\'a}.
\newblock Finding small stabilizers for unstable graphs.
\newblock {\em Mathematical Programming}, 154:173--196, 2015.

\bibitem{brand2016claw}
A.~Brandstaett and R.~Mosca.
\newblock Maximum weight independent set in $l$claw-free graphs in polynomial
  time.
\newblock {\em Discrete Applied Mathematics}, 237:57--64, 2018.

\bibitem{CEW11}
G.~Chalkiadakis, E.~Elkind, and M.~Wooldridge.
\newblock {\em Computational Aspects of Cooperative Game Theory}.
  Morgan and Claypool Publishers, 2011.

\bibitem{DW06}
V.~G. Deineko and G.~J. Woeginger.
\newblock On the dimension of simple monotonic games.
\newblock {\em European Journal of Operational Research}, 170(1):315--318,
  2006.
  
\bibitem{Edmonds70}
J. Edmonds and D.R. Fulkerson.
\newblock Bottleneck extrema.
\newblock {\em Journal of Combinatorial Theory}, 8(3):299 - 306,
1970.

\bibitem{ECJ08}
E.~Elkind, G.~Chalkiadakis, and N.~R. Jennings.
\newblock Coalition structures in weighted voting games.
\newblock volume 178, pages 393--397, 2008.

\bibitem{EGGW09}
E.~Elkind, L.~A. Goldberg, P.~W. Goldberg, and M.~Wooldridge.
\newblock On the computational complexity of weighted voting games.
\newblock {\em Annals of Mathematics and Artificial Intelligence},
  56(2):109--131, 2009.

\bibitem{faigle1998}
U.~Faigle, W.~Kern, S.~Fekete, and W.~Hochstaettler.
\newblock The nucleon of cooperative games and an algorithm for matching games.
\newblock {\em Mathematical Programming}, 83:195--211, 1998.

\bibitem{paper_alpha_roughly_weighted}
J.~Freixas and S.~Kurz.
\newblock On $\alpha$-roughly weighted games.
\newblock {\em International Journal of Game Theory}, 43(3):659--692, 2014.

\bibitem{freixas2011complexity}
J.~Freixas, X.~Molinero, M.~Olsen, and M.~Serna.
\newblock On the complexity of problems on simple games.
\newblock {\em {RAIRO}-Operations Research}, 45(4):295--314, 2011.

\bibitem{FP08}
J.~Freixas and M.~A. Puente.
\newblock Dimension of complete simple games with minimum.
\newblock {\em European Journal of Operational Research}, 188(2):555--568,
  2008.

\bibitem{GJ79}
M.~R. Garey and D.~S. Johnson.
\newblock {\em Computers and Intractability: A Guide to the Theory of
  NP-Completeness}.
\newblock W. H. Freeman \& Co., New York, NY, USA, 1979.

\bibitem{gvozdeva2013three}
T.~Gvozdeva, L.~A. Hemaspaandra, and A.~Slinko.
\newblock Three hierarchies of simple games parameterized by
  {\lq\lq}resource{\rq\rq} parameters.
\newblock {\em International Journal of Game Theory}, 42(1):1--17, 2013.

\bibitem{hegedus1996geometric}
T.~Heged{\"u}s and N.~Megiddo.
\newblock On the geometric separability of {B}oolean functions.
\newblock {\em Discrete Applied Mathematics}, 66(3):205--218, 1996.

\bibitem{hof2016weight}
F.~Hof.
\newblock Weight distribution in matching games.
\newblock MSc Thesis, University of Twente, 2016.

\bibitem{HKKP18}
F. Hof, W. Kern, S. Kurz and D. Paulusma, Simple games versus weighted voting games, Proc. SAGT 2018, LNCS 11059 (2018) 69--81.

\bibitem{Is56}
J.~R. Isbell.
\newblock A class of majority games.
\newblock {\em Quarterly Journal of Mathematics}, 7:183--187, 1956.

\bibitem{kern2003matching}
W.~Kern and D.~Paulusma.
\newblock Matching games: The least core and the nucleolus.
\newblock {\em Mathematics of Operations Research}, 28:294--308, 2003.

\bibitem{KKT}
J.~Koenemann, K.~Pashkovich, and J.~Toth.
\newblock Computing the nucleolus of weighted cooperative matching games in
  polynomial time.
  \newblock arXiv:1803.03249v2, 9 March 2018.

\bibitem{high_dimension}
S.~Kurz, X.~Molinero, and M.~Olsen.
\newblock On the construction of high dimensional simple games.
\newblock In Proc. ECAI 2016,
pages 880--885, New York, 2016.

\bibitem{lovasz2009matching}
L.~Lov{\'a}sz and M.~D. Plummer.
\newblock {\em Matching theory}, volume 367.
\newblock American Mathematical Society, 2009.

\bibitem{Pa18a}
K. Pashkovich, Computing the nucleolus of weighted voting games in pseudo-polynomial time, arXiv:1810.02670.

\bibitem{P18}
K. Pashkovich.
\newblock On critical threshold value for simple games.
\newblock arXiv:1806.03170v2, 11 June 2018.

\bibitem{peled1985polynomial}
U.~N. Peled and B.~Simeone.
\newblock Polynomial-time algorithms for regular set-covering and threshold
  synthesis.
\newblock {\em Discrete Applied Mathematics}, 12(1):57--69, 1985.

\bibitem{peters2008games}
H.~Peters.
\newblock {\em Game Theory}.
\newblock Springer, 2008.


\bibitem{Sc00}
A.~Schrijver.
\newblock A combinatorial algorithm minimizing submodular functions in strongly
  polynomial time.
\newblock {\em Journal of Combinatorial Theory, Series {B}}, 80(2):346--355, 2000.

\bibitem{SR94}
T.~Solymosi and T.~E. Raghavan.
\newblock An algorithm for finding the nucleolus of assignment games.
\newblock {\em International Journal of Game Theory}, 23:119--143, 1994.

\bibitem{TZ93}
A.~D. Taylor and W.~S. Zwicker.
\newblock Weighted voting, multicameral representation, and power.
\newblock {\em Games and Economic Behavior}, 5:170--181, 1993.

\bibitem{taylor1999simple}
A.~D. Taylor and W.~S. Zwicker.
\newblock {\em Simple games: Desirability relations, trading,
  pseudoweightings}.
\newblock Princeton University Press, 1999.

\bibitem{TIAS77}
S.~Tsukiyama, M.~Ide, H.~Ariyoshi, and I.~Shirakawa.
\newblock A new algorithm for generating all the maximal independent sets.
\newblock {\em {SIAM} Journal on Computing}, 6(3):505--517, 1977.
\end{thebibliography}

  \end{document}